\documentclass[journal]{IEEEtran}
\usepackage{amsmath,amsfonts,amsthm}

\usepackage{color}

\usepackage{float}
\usepackage{array}
\usepackage[caption=false,font=normalsize,labelfont=sf,textfont=sf]{subfig}
\usepackage{textcomp, booktabs, multirow}
\usepackage{stfloats,enumerate,enumitem}
\usepackage{url}
\usepackage{verbatim}
\usepackage{graphicx}
\usepackage{cite}
\newtheorem{theorem}{Theorem}
\newtheorem{lemma}{Lemma}

\newtheorem{corollary}{Corollary}

\usepackage{algorithm, algorithmic}

\def\E{\mathbb{E}}
\def\P{\mathbb{P}}

\newcommand{\premu}{\boldsymbol{\mu}_0}
\newcommand{\presigma}{\boldsymbol{\Sigma}_0}
\newcommand{\postmu}{\boldsymbol{\mu}_1}
\newcommand{\postsigma}{\boldsymbol{\Sigma}_1}

\newcommand{\hpostsigma}{\widehat{\boldsymbol{\Sigma}}_1}

\newcommand{\grid}{\mathcal{G}}

\newcommand{\x}{\mathbf{x}}
\newcommand{\nx}{\tilde{\mathbf{x}}}
\newcommand{\dime}{p}
\newcommand{\noise}{\mathbf{e}}
\newcommand{\noisedev}{\sigma_e}
\newcommand{\noisesigma}{\sigma_e^2}
\newcommand{\sens}{\Delta}

\newcommand{\mnoisesigma}{\boldsymbol{D}_e}
\newcommand{\eigenpre}{\nu_0}
\newcommand{\eigenpost}{\nu_1}
\newcommand{\FAR}{\text{FAR}}
\newcommand{\ADD}{\underline{\text{ADD}}}
\usepackage{xspace}
\newcommand{\Abbe}{PLOD\xspace}

\begin{document}

\title{
\vspace{-0.9em}
Privacy-Preserving Line Outage Detection in Distribution Grids: An Efficient Approach with Uncompromised Performance}
\author{Chenhan Xiao,~\IEEEmembership{Student Member,~IEEE}, Yizheng Liao,~\IEEEmembership{Member,~IEEE}, Yang Weng,~\IEEEmembership{Senior Member,~IEEE}
\vspace{-2.5em}
}
\maketitle
\begin{abstract}
Recent advancements in research have shown the efficacy of employing sensor measurements, such as voltage and power data, in identifying line outages within distribution grids. However, these measurements inadvertently pose privacy risks to electricity customers by potentially revealing their sensitive information, such as household occupancy and economic status, to adversaries. To safeguard raw data from direct exposure to third-party adversaries, this paper proposes a novel decentralized data encryption scheme. The effectiveness of this encryption strategy is validated via demonstration of its differential privacy attributes by studying the Gaussian differential privacy. Recognizing that the encryption of raw data could affect the efficacy of outage detection, this paper analyzes the performance degradation by examining the Kullback–Leibler divergence between data distributions before and after the line outage. This analysis allows us to further alleviate the performance degradation by designing an innovative detection statistic that accurately approximates the optimal one. Manipulating the variance of this statistic, we demonstrate its ability to approach the optimal detection performance. The proposed privacy-aware detection procedure is evaluated using representative distribution grids and real load profiles, covering 17 distinct outage configurations. Our empirical results confirm the privacy-preserving nature of our approach and show that it achieves comparable detection performance to the optimal baseline.
\end{abstract}

\vspace{-1.5em}
\section{Introduction}
In distribution grids, the detection of line outages is essential for system monitoring and control, playing a critical role in the restoration of network stability and the mitigation of customer losses. According to the U.S. Energy Information Administration \cite{USreport}, customers experienced over seven hours of power interruptions in 2021, attributed mainly to severe weather events and power supply shortages. Traditionally, utility companies have installed smart meters with Advanced Metering Infrastructure (AMI) and Fault Location, Isolation, and Service Restoration (FLISR) systems to report outages in cases of power absence \cite{location2014isolation}. However, these ``last gasp'' notifications are limited when customers continue to have power after the line outage, from distributed energy resources such as rooftop solar panels, battery storage, and electric vehicles, which are now widely adopted. Additionally, in some urban areas, secondary distribution grids are mesh networks. In this setup, a single line outage induced by circuit faults or human interference may not result in a power outage because of alternative power supply routes. Consequently, smart meters at customer end also cannot report outages.

To identify these types of line outages, real-time sensor measurements, including voltage magnitudes, phasor angles, and load estimates, have been employed and confirmed for their effectiveness \cite{he2010fault, babakmehr2019compressive, sevlian2017outage, soleymani2019unsupervised, liao2021quick, liao2022quickest}. However, the utilization of real-time sensor measurements raises privacy concerns, particularly regarding the potential exposure of sensitive information. For example, if a customer's time-series grid data were provided to an untrusted third party, they could deduce appliance usage \cite{zoha2012non} and unveil details about household occupancy and economic status (as illustrated in the lower half of Fig. \ref{fig:bigpic}) using non-intrusive load monitoring techniques \cite{wood2003dynamic, mcdaniel2009security}. Therefore, it is crucial to safeguard such data against direct disclosure to third parties during the outage detection process.

In pursuing a privacy-aware outage detection procedure, we choose to develop a decentralized randomization scheme based on a probabilistic methodology for encrypting the raw data. Among the methodologies for utilizing sensor measurements in outage detection, both deterministic \cite{sevlian2017outage, babakmehr2019compressive} and probabilistic \cite{dwivedi2021scalable, xiao2023distribution} approaches have been proposed. Deterministic methods typically set a threshold and declare an outage when data changes exceed this threshold. Although these techniques are easy to implement, they do not align with our concept of a randomization scheme for data encryption. In contrast, probabilistic approaches focus on monitoring changes in the probability distribution of sensor measurements, providing a suitable foundation for our approach. The core idea is to alter the absolute values of end-user measurements to protect end-user privacy while preserving the relative changes in data distribution before and after an outage event (see upper half of Fig. \ref{fig:bigpic}). We also want to point out that while community-level data, aggregated from individual users, can inherently safeguard end-user privacy, it complicates the detection of line outages in large distribution grids and renders the precise localization of outage branches unattainable.

Specifically, we aim to develop a privacy-aware outage detection procedure based on our prior research \cite{liao2021quick, xiao2023distribution}, which utilizes a probabilistic change point detection (CPD) method known for its guaranteed performance. The CPD approach is adopted for detecting changes in the probabilistic distribution of sensor measurements while adhering to a predefined false alarm tolerance constraint \cite{shiryaev1963optimum}. In our problem, the sensor measurements are modeled as a time-series data stream $\x[n] \in \mathbb{R}^{\dime}$, where $n \in \mathbb{N}$ corresponds to the time step. These time-series data are assumed to exhibit distinct probabilistic distributions before and after an outage time $\lambda \in \mathbb{N}$:
\begin{equation}
\label{eq:sequence}
\x[n] \stackrel{i.i.d}{\sim} g, \ n < \lambda \quad \text{and} \quad \x[n] \stackrel{i.i.d}{\sim} f, \ n \ge \lambda,
\end{equation}
where $g$ and $f$ represent the distributions before and after the outage, respectively. The CPD framework with sensor data defined in (\ref{eq:sequence}) has been applied to detect line outages and faults in transmission grids  \cite{Chen2016Quick} as well as in DC micro-grids  \cite{gajula2021quickest}. These applications benefit from theoretical guarantees regarding optimal detection delay, as studied in \cite{tartakovsky2005general}.

In addition to detecting power line outages using sensor data from electricity customers, many other applications of the CPD framework involve similar privacy concerns related to the use of sensitive data. Such applications include monitoring patient health based on heart rates \cite{yang2006adaptive} and evaluating financial conditions using transaction data \cite{hand2001prospecting}. Consequently, the development of a privacy-aware CPD that preserves its detection performance has emerged as a substantial area of interest and is the primary focus of this paper.

To safeguard privacy, recent studies have introduced randomization schemes to encrypt data, effectively concealing sensitive information from potential attackers. In assessing the level of privacy achieved by such randomization schemes, the differential privacy framework \cite{dwork2006calibrating} is employed, offering a worst-case privacy guarantee. In the context of parametric CPD, where distributions $g$ and $f$ are known in (\ref{eq:sequence}), \cite{cummings2018differentially} utilized noisy approximation algorithms developed by \cite{dwork2014algorithmic} to compute a privately approximated change-point maximum likelihood estimation. In non-parametric CPD scenarios where the distributions $g$ and $f$ are unknown, \cite{cummings2020privately} privately estimated the change points using the Mann-Whitney test \cite{wilcoxon1945individual}. These studies involved encrypting the detection statistic with Laplace noise after a trusted third party collected the raw data $\x[n]$. In cases where a trusted third party is absent, \cite{berrett2021locally} proposed randomizing the raw data with Laplace noise, ensuring that the raw data remains inaccessible to anyone except its original holder. Despite the privacy guarantees offered by existing randomization approaches, there remain several limitations due to the complexity of the privacy mechanisms or the intricate structure of the data. First, to pursue the differential privacy framework, many existing works \cite{cummings2018differentially, cummings2020privately} choose to apply noise to statistic-level after raw data is collected, potentially exposing raw data to breaches before encryption. Second, many existing works lack a rigorous quantification of how privacy mechanisms affect detection performance \cite{lau2020privacy,cummings2018differentially, cummings2020privately,lau2020quickest}. Third, despite the advances in privacy protection, there is little exploration of methods to mitigate the negative effects of these privacy measures on system functionality \cite{zhang2021single, aminikhanghahi2017survey}. To the best of our knowledge, safeguarding privacy without compromising detection performance remains out of the reach of existing theory.

In this paper, we narrow our focus on the parametric setting of CPD for line outage detection. Having knowledge of the distributions $g$ and $f$ allows us to quantify the cost associated with introducing privacy guarantees into the outage detection procedure. Furthermore, it empowers us to design a novel detection statistic aimed at mitigating this cost. To secure the privacy at user-level, our first innovation is developing a decentralized encryption scheme directly on raw data. Unlike existing work \cite{cummings2018differentially,cummings2020privately} that introduced noise to statistics after raw data is collected, our approach ensures user data's confidentiality before any external access. To demonstrate that this scheme adheres to classic differential privacy, we detour the proof through Gaussian differential privacy \cite{dong2019gaussian}, an extension of differential privacy applicable to arbitrary distributions.

Despite the privacy guarantee, there is an inevitable compromise in detection performance due to the encryption of data. Our second contribution is to quantify the extent of performance compromise in pursuit of varying privacy levels. By investigating the Kullback–Leibler divergence between distributions $f$ and $g$, we pinpoint how encryption-induced noises extend the outage detection delays. Our findings provide a foundational framework for assessing the implications of privacy-enhancing technologies on operational capabilities.

In our third contribution, we tackle the challenge of performance degradation due to privacy measures by devising a novel detection statistic. This statistic innovatively estimates the optimal statistic achievable with raw data, minimizing detection delays and respecting false alarm constraints. Our analysis, rooted in Jensen's inequality, reveals that controlling the variance of this statistic significantly narrows the performance gap. By strategically reducing the variance, we demonstrate the potential to virtually eliminate the adverse impacts of privacy protection on detection performance, marking a significant advancement in the field.

In summary, our contributions include: (1) We innovate by introducing noise directly into end-user-level data while ensuring adherence to the differential privacy framework. It enhances privacy beyond existing methods that only add noise to statistic-level after end-user-level data is collected. (2) We rigorously quantify the impact of privacy-induced noise on outage detection performance using Kullback-Leibler divergence, providing a detailed analysis of the trade-off between privacy and efficiency. (3) We propose a novel noise-mitigation technique that significantly reduces the negative effects of privacy protection on detection accuracy, achieving near-optimal performance levels. To validate our contributions, we conduct comprehensive experiments utilizing representative distribution grids and real load profiles, covering 17 distinct outage configurations.

In the following, Section \ref{sec:pre} introduces the preliminary aspects of our system modeling, the CPD framework, and the differential privacy framework. Section \ref{sec:method} presents our privacy-aware detection procedure that dose not compromise detection performance. Section \ref{sec:simulation} assesses our method using four distribution grids and real-world load profiles. Section \ref{sec:conclusion} concludes of this paper.

\begin{figure*}[ht]
\centering
\includegraphics[width=1\linewidth]{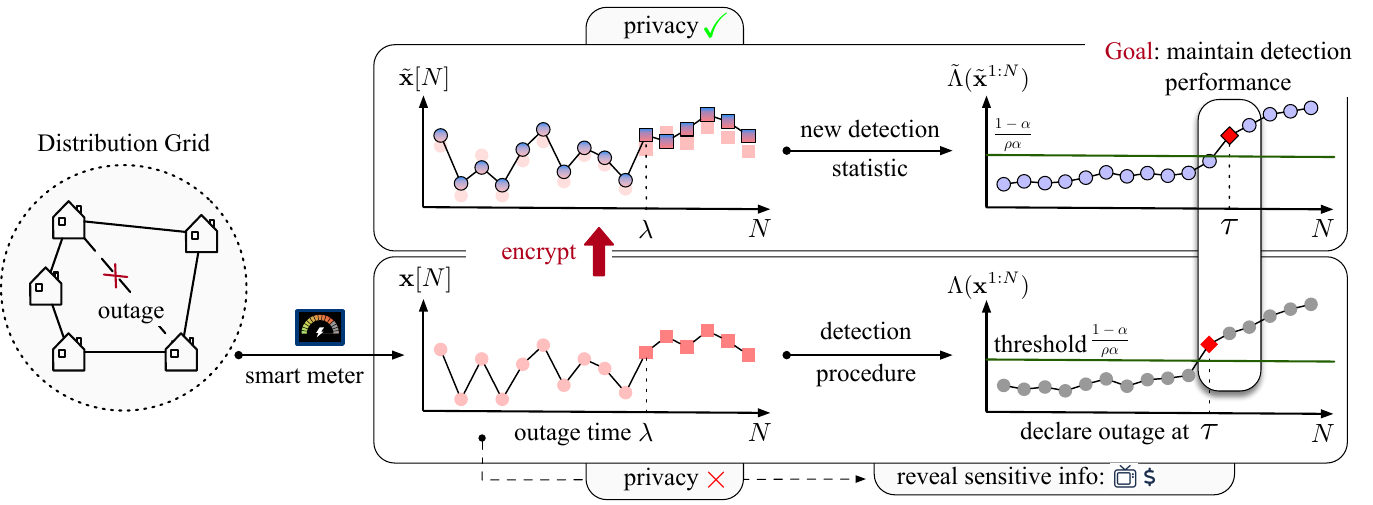}
\vskip -0.1in
\caption{An overview of the privacy-aware line outage detection problem in the distribution grid.}
\vskip -0.15in
\label{fig:bigpic} 
\end{figure*}

\vspace{-1.2em}
\section{Preliminary}
\label{sec:pre}
{\bf System Modeling.}\ To illustrate our probabilistic design for the privacy-aware detection procedure, we define the following variables. The voltage magnitude at each bus $i \in \grid$ is modeled as a random variable $V_i$, where $\grid:=\{1,2,\cdots,\dime\}$ represents the distribution grid as a graph containing $\dime>0$ buses. At time step $n$, we denote the realization of $V_i$ as $v_i[n] \in \mathbb{R}$ in per unit, and we use $\mathbf{v}[n]=\{v_1[n],\cdots,v_{\dime}[n]\}\in \mathbb{R}^{\dime}$ to represent the collection of voltage magnitudes in the grid $\grid$. Finally, we use the notation $\x[n] = \mathbf{v}[n] - \mathbf{v}[n-1]$ to denote the incremental change in voltage magnitudes.\footnote{For simplicity, we use the notation $\x$ instead of $\Delta \mathbf{v}$.} 

We utilize voltage increment data because \cite{xiao2023distribution} establishes that this data adheres to two multivariate Gaussian distributions, denoted as $g\sim\mathcal{N}(\premu,\presigma)$ and $f\sim\mathcal{N}(\postmu,\postsigma)$ before and after a line outage. For the sake of simplicity, we also use the notation $\x^{1:N} = \{\x[1],\cdots,\x[N]\}$ to represent all the measurements up to time $N$.

Based on the modeling, the problem of detecting distribution grid line outages while preserving privacy is formally defined as follows (refer to Fig. \ref{fig:bigpic} for visualization):
\begin{itemize}[leftmargin=*]
    \item {\bf Given}: A stream of voltage magnitude increments $\x^{1:N}$ from the smart meters.
    \item {\bf Find}: The line outage time $\lambda$ as quickly as possible.
    \item {\bf Require:} Avoid disclosing the raw data $\x^{1:N}$.
\end{itemize}

{\bf Outage detection based on change point detection.} To detect the outage time $\lambda$ in (\ref{eq:sequence}) using voltage magnitude increments $\x^{1:N}$, our previous work \cite{liao2021quick,xiao2023distribution} follows the Bayesian detection procedure  \cite{shiryaev1963optimum,tartakovsky2005general}. That is, identifying the outage time is equivalent to performing the hypothesis test:
$$
   \mathcal{H}_0:\lambda>N \quad\text{and}\quad \mathcal{H}_1:\lambda\leq N
$$
sequentially given data $\x^{1:N} = \{\x[1],\cdots ,\x[n],\cdots,\x[N]\}$. As data is received in a streaming manner ($N$ increases), the first time hypothesis $\mathcal{H}_0$ is rejected reveals the value of $\lambda$. To determine when to reject $\mathcal{H}_0$, the posterior probability ratio
\begin{align}\label{eq:statistic}
\Lambda(\x^{1:N})=\frac{\P(\lambda\leq n|\x^{1:N})}{\P(\lambda>n|\x^{1:N})} = \sum_{k=1}^N\pi_N^k\prod_{n=k}^{N}\frac{f(\x[n])}{g(\x[n])}
\end{align}
is calculated at each time step $N$. $\lambda\in\mathbb{N}$ is assumed to follow a prior distribution $\pi$, and we define $\pi_N^k=\frac{\pi(k)}{\sum_{k=N+1}^\infty \pi(k)}$ for simplicity. The ratio in (\ref{eq:statistic}) compares the probabilities of ``outage occurred ($\lambda\leq N$)'' and ``outage did not occur ($\lambda> N$)'' given the historical measurements $\x^{1:N}$. A larger ratio indicates that ``outage occurred'' is more likely than ``outage did not occur''. Therefore, we declare the outage time $\lambda$ when the ratio in (\ref{eq:statistic}) exceeds a predefined threshold.
By the Shiryaev-Roberts-Pollaks procedure \cite{shiryaev1963optimum,tartakovsky2005general}, the following threshold in Theorem \ref{theorem:Bayesian} optimally considers the trade-off between the false alarm and the detection delay.

\begin{theorem}\label{theorem:Bayesian}
When $\lambda$ follows a geometric prior $\text{Geo}(\rho)$, we declare the outage time when the posterior probability ratio $\Lambda(\x^{1:N})$ exceeds the threshold $\frac{1-\alpha}{\rho\alpha}$ for the first time as
\begin{equation}\label{eq:stop rule}
\tau=\inf\{N\in\mathbb{N}:\Lambda(\x^{1:N})\geq \frac{1-\alpha}{\rho\alpha}\}.
\end{equation}
\end{theorem}

The detection procedure (\ref{eq:stop rule}) constrains that the false alarm rate (FAR) remains below a pre-defined tolerance level $\alpha$, i.e., $\FAR(\Lambda,f,g)\overset{\Delta}{=}\mathbb{P}(\tau<\lambda)\le \alpha$. More importantly, as $\alpha\to0$, $\tau$ is asymptotically optimal for minimizing the average detection delay (ADD) as
\begin{align}\label{eq:lowerbound}
&\E[\tau-\lambda|\tau\geq\lambda] =\inf_{\P(\tau^{\ast}\leq\lambda)\leq\alpha}\E[\tau^{\ast}-\lambda|\tau^{\ast}\geq\lambda]\nonumber\\
&= \frac{|\log\alpha|}{-\log(1-\rho)+D_{\text{KL}}(f||g)}\overset{\Delta}{=}\ADD(\Lambda,f,g),
\end{align}
where $D_{\text{KL}}(f\|g)$ denotes the Kullback–Leibler (KL) divergence between distributions $f$ and $g$.

In the practical application of outage detection, the determination of the threshold $\alpha$ in (\ref{eq:stop rule}) is achieved through a systematic and iterative methodology, which is anchored in both statistical analysis and operational considerations. Initial selection is based on analyzing sensor data $\x^{1:N}$ to distinguish between normal variations and potential outages, taking into account the balance between minimizing detection delays and reducing false alarms. This choice is refined through iterative testing with historical data, allowing us to fine-tune $\alpha$ to optimize detection accuracy while considering the operational impact of false positives. Moreover, we incorporate the flexibility to adjust $\alpha$ dynamically, accommodating seasonal variations and evolving grid conditions, ensuring our algorithm remains effective and reliable across different scenarios.

After detecting the occurrence of a line outage, accurately localizing the affected branch is vital for system restoration. In \cite{liao2021quick}, researchers introduced a precise outage localization method by demonstrating the conditional independence of voltage increments between two disconnected buses. Their technique involved computing the conditional correlation between all potential bus pairs in the grid and detecting changes from non-zero to zero values. Unlike methods relying on nodal electric circuit analysis for fault location estimation, this approach offers a distinct method that relies solely on the covariance matrices of the voltage data. This attribute is shown later to be efficient even in privacy-aware contexts.

To estimate the conditional correlation between bus $i$ and bus $k$, the covariance matrix $\boldsymbol{\Sigma}$ is utilized. Let set $\mathcal{I}:=\{i,k\}$ and $ \mathcal{K}:=\grid\backslash\{i,k\}$, the covariance matrix is decomposed as $\boldsymbol{\Sigma} = \left[ \begin{matrix}\boldsymbol{\Sigma}_{\mathcal{I}\mathcal{I}} & \boldsymbol{\Sigma}_{\mathcal{I}\mathcal{K}}\\ \boldsymbol{\Sigma}_{\mathcal{I}\mathcal{K}}^\top & \boldsymbol{\Sigma}_{\mathcal{K}\mathcal{K}}\end{matrix}\right]$. Based on this, the conditional correlation $\rho_{ik}$ between bus $i$ and bus $k$ is
\begin{equation}\label{eq:conditional}
    \rho_{ik}(\boldsymbol{\Sigma}) = \frac{\boldsymbol{\Sigma}_{\mathcal{I}|\mathcal{K}}(1,2)}{\sqrt{\boldsymbol{\Sigma}_{\mathcal{I}|\mathcal{K}}(1,1)\boldsymbol{\Sigma}_{\mathcal{I}|\mathcal{K}}(2,2)}},
\end{equation}
where the conditional covariance is computed by the Schur complement \cite{boyd2004convex} as $\boldsymbol{\Sigma}_{\mathcal{I}|\mathcal{K}} = \boldsymbol{\Sigma}_{\mathcal{I}\mathcal{I}} - \boldsymbol{\Sigma}_{\mathcal{I}\mathcal{K}}\boldsymbol{\Sigma}_{\mathcal{K}\mathcal{K}}^{-1}\boldsymbol{\Sigma}_{\mathcal{I}\mathcal{K}}^\top$.

\begin{theorem}(Line outage localization). \label{lemma:local}
The conditional correlation is calculated based on (\ref{eq:conditional}) for every pair of $(i,k)$ as \vspace{-0.5em}
\begin{equation}
    \underbrace{\rho_{ik}^{-} = \rho_{ik}(\presigma)}_{\text{before outage}}\quad\text{and}\quad \underbrace{\rho_{ik}^{+} = \rho_{ik}(\hpostsigma)
    }_{\text{after outage}}.\vspace{-0.5em}
\end{equation}
The branch between bus $i$ and $k$ is out-of-service if
$
|\rho_{ik}^{-}| > \delta_{\max}$ and $|\rho_{ik}^{+}| < \delta_{\min}$. The thresholds are set as $\delta_{\max}=0.5$ and $\delta_{\min}=0.1$ based on real-world outage data to check if the correlation changes from non-zero to near-zero value.
\end{theorem}

According to Theorem \ref{lemma:local}, we track the change of covariance matrices to localize the out-of-service branch. Specifically, an out-of-service branch between bus $i$ and bus $k$ can be identified if both of the following conditions are met simultaneously: (1) $|\rho_{ik}^{-}| > \delta_{\max}$ indicating the presence of a branch between buses $i$ and $k$ before the outage, and (2) $|\rho_{ik}^{+}| < \delta_{\min}$ indicating the absence of a branch between buses $i$ and $k$ after the outage.

{\bf Differential privacy.} To assess the level of privacy preservation, we follow the framework of differential privacy \cite{dwork2006calibrating}, which offers worst-case privacy guarantees. Specifically, an algorithm $\mathcal{M}:\mathbb{R}^{\dime}\to\mathbb{R}^\dime$ is $(\varepsilon,\delta)$-differentially private if, for any neighboring datasets $X$ and $X^\prime$ (differing in at most one element), and for every subset of possible outputs $\mathcal{S}$, the following inequality holds:
\begin{align}\label{eq:dp}
\P[\mathcal{M}(X)\in \mathcal{S}] \le \exp(\varepsilon)\P[\mathcal{M}(X^\prime)\in \mathcal{S}] + \delta.
\end{align}
In essence, this property ensures that a potential attacker observing the outcomes of the algorithm $\mathcal{M}$ cannot easily deduce whether a specific individual's information is present in the dataset. While the conventional technique for achieving differential privacy involves the introduction of Laplace noise \cite{cummings2018differentially} to raw data, the concept of Gaussian differential privacy \cite{dong2019gaussian}  extends differential privacy to encompass noises generated from a broader range of distributions.

\section{Privacy-aware Line Outage Detection with Boosted Performance}
\label{sec:method}

In the aforementioned outage identification procedure (\ref{eq:stop rule}), the increments of voltage magnitude data $\x^{1:N}$ are critical. However, such data may also be used to infer customer's sensitive information \cite{wood2003dynamic,mcdaniel2009security}, such as the household occupancy (see lower half of Fig. \ref{fig:bigpic}), e.g., when the house owner arrives or leaves home. To protect the raw voltage data of customers, at each time step $n$ when data $\x[n]$ is received, we apply a randomizing scheme to encrypt the raw data directly:
\begin{equation}\label{eq:addnoise}
    \nx[n] = \x[n] + \noise[n],
\end{equation}
where $\noise[n]\in\mathbb{R}^{\dime}$ is a random noise vector. The randomized approach stands out for its simplicity and effectiveness over other techniques like Homomorphic Encryption or Data Anonymization. Its introduction of systematic noise $\noise[n]$ not only facilitates differential privacy, offering a measurable level of privacy protection, but also enables the quantification of any impact on detection performance. This dual capability allows for a finely tuned balance between ensuring user privacy and maintaining the accuracy of outage detection efforts.

The noise $\noise[n]$ has to be sufficiently large to hide the characteristics of the raw data while not being too large to impact the detection performance. To establish the suitable level of embedded noise, we evaluate the privacy guarantee within the differential privacy framework in Section \ref{sec:privacy} and assess the corresponding degradation in detection performance in Section \ref{sec:degradation}. We show that, in general, the noise added to data makes it harder to distinguish whether the data comes from the distribution $g$ or $f$, leading to a prolonged detection delay. Integrating these analyses, we propose a new statistic in Section \ref{sec:statistic} (to replace (\ref{eq:statistic})) such that the new detection procedure is both privacy-preserving and has comparable detection performance as the optimal case with access to raw data.

For making the randomizing scheme (\ref{eq:addnoise}) satisfy the differential privacy, we generate noise from the same distribution (Gaussian) as raw data, i.e., $\noise[n]\sim\mathcal{N}(\boldsymbol{0},\mnoisesigma)$. The covariance matrix is designed to be diagonal, i.e., $\boldsymbol{D}_e=\text{diag}(\noisesigma,\cdots,\noisesigma)$ where variance $\noisesigma$ represents the noise level or amount of noise. A diagonal covariance indicates that each element in the noise vector is independent. In doing so, the scheme (\ref{eq:addnoise}) is equivalent to adding a random noise scalar to each dimension of the data vector, ensuring that each customer's raw data is encrypted before sending to any third party (see Fig. \ref{fig:addnoise}). Notice, unlike some works that add noise to the statistics (e.g., $\Lambda(\x^{1:N})$) \cite{cummings2018differentially,cummings2020privately} after raw data is collected, our approach ensures no direct exposure of the raw data.

There is another advantage of using a diagonal noise covariance $\boldsymbol{D}_e$, i.e., introducing independent noise into user data. In fact, this choice allows us to effectively differentiate line outages from other causes of voltage distribution changes by examining the voltage data's covariance matrix $\boldsymbol{\Sigma}$. As detailed in our previous work \cite{xiao2023distribution}, voltage increments between disconnected buses show conditional Independence. Specifically, if the branch connecting two buses becomes non-operational, the conditional correlation between these buses shifts from a non-zero value to zero, highlighting a unique pattern of line outages in the voltage data's covariance structure. More importantly, our privacy-preserving technique, when introducing independent noise to each bus's data, does not compromise this unique property, thus maintaining the ability to differentiate line outages effectively.

\begin{figure}[ht]
\centering
\vskip -0.15in
\includegraphics[width=0.9\linewidth]{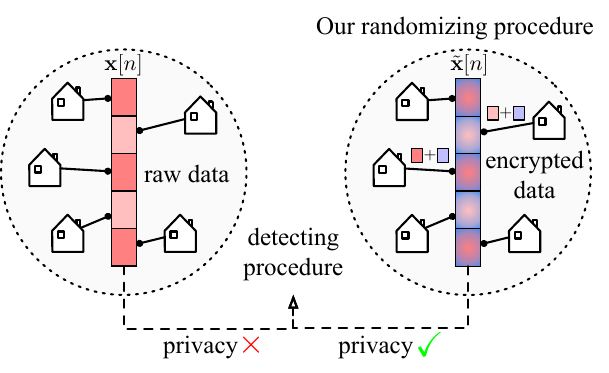}
\vskip -0.15in
\caption{The decentralized randomizing scheme (\ref{eq:addnoise}) to protect privacy of each customer $i$'s raw data $\x_i[n]$ in the data vector $\x[n]$.}
\vskip -0.2in
\label{fig:addnoise}
\end{figure}

\subsection{Differential Privacy Guarantee of the Randomizing Scheme}
\label{sec:privacy}
Applying the randomizing scheme (\ref{eq:addnoise}), the detection procedure will be performed on the encrypted data $\nx^{1:N}=\{\nx[1],\cdots,\nx[N]\}$ to find the outage (see Fig. \ref{fig:addnoise}). In this subsection, we quantify how much privacy is preserved w.r.t. the noise level $\noisesigma$. To achieve this, we prove that (\ref{eq:addnoise}) satisfies the classic $(\varepsilon,\delta)$-differential privacy mechanism \cite{dwork2006calibrating}. 

A differential privacy scheme indicates that by looking at the encrypted data $\nx[n]$, an adversary struggles to tell whether any piece of real data $\x_i[n]$ is included. The mathematical definition is given in (\ref{eq:dp}).
Since noise $\noise[n]$ is independent Gaussian to raw data $\x[n]$, the encrypted data $\nx[n]$ also follows Gaussian. It allows us to detour the proof of classic differential privacy by the tool of Gaussian differential privacy \cite{dong2019gaussian}. Specifically, a $G_{\mu}$-Gaussian differential privacy scheme implies the following: telling whether any piece of real data $\x_i[n]$ is present in the encrypted data $\nx[n]$ is more difficult than distinguishing between distributions $\mathcal{N}(0,1)$ and $\mathcal{N}(\mu,1)$. The difficulty is quantified using the trade-off function $T(\mathcal{N}(0,1), \mathcal{N}(\mu,1))$, which characterizes the balance between type I and type II errors in distinguishing these distributions \cite{dong2019gaussian}. This particular trade-off function is also referred to as $G_{\mu}$. In Lemma \ref{the:g-dp}, we show that our scheme (\ref{eq:addnoise}) is Gaussian differential private.

\begin{lemma}\label{the:g-dp}
The randomizing scheme (\ref{eq:addnoise}) is $G_{\frac{\sens}{\noisedev}}$-Gaussian differential private \cite{dong2019gaussian} where $\sens:= \sup_{\x[n], \x^\prime[n]} \|\x[n] - \x^\prime[n]\|$ is the sensitivity of raw data, and $\x[n], \x^\prime[n]$ only differs in exactly one element.
\end{lemma}

\begin{proof}
The encrypted data $\nx[n]$ and its neighboring data $\nx^{\prime}[n]$ (i.e., they differ in exactly one element) both follow Gaussian distributions as 
$\nx[n]\sim \mathcal{N}(\x[n],\boldsymbol{D}_e)$ and $\nx^{\prime}[n]\sim \mathcal{N}(\x^\prime[n],\boldsymbol{D}_e)$. Then, we have
\begin{align}
    T\left(\nx[n], \nx^\prime[n]\right) &= T(\mathcal{N}(\x[n],\boldsymbol{D}_e), \mathcal{N}(\x^{\prime}[n],\boldsymbol{D}_e))\nonumber\\
    &= G_{\|\x[n] - \x^\prime[n]\|/\noisedev} \ge G_{\frac{\sens}{\noisedev}}, \label{eq:G-dp}
\end{align}
where $T(\nx[n], \nx^\prime[n])$ is defined as the trade-off function between type I and II errors in differentiating data $\nx[n]$ and $\nx^\prime[n]$. The inequality is due to the definition of sensitivity, i.e., $\|(\x[n] - \x^\prime[n])/\noisedev\|\le \frac{\sens}{\noisedev}$. 
\end{proof}

Given the foundation of Gaussian differential privacy, we are ready to demonstrate that our scheme (\ref{eq:addnoise}) also adheres to the classic $(\varepsilon,\delta)$-differential privacy \cite{dwork2006calibrating}.

\begin{corollary}\label{co:dp}
Provided the $G_{\frac{\sens}{\noisedev}}$-Gaussian differential privacy, (\ref{eq:addnoise}) satisfies the $(\varepsilon,\delta(\varepsilon))$-differential privacy \cite{dong2019gaussian} where
$$
\delta(\varepsilon) = \Phi(-\frac{\varepsilon\noisedev}{\sens}+\frac{\sens}{2\noisedev}) - e^{\varepsilon}\Phi(-\frac{\varepsilon\noisedev}{\sens}-\frac{\sens}{2\noisedev}),
$$
and $\Phi$ is the CDF of the unit normal distribution.
\end{corollary}

Satisfying the $(\varepsilon,\delta(\varepsilon))$-differential privacy in Corollary \ref{co:dp}, our proposed scheme (\ref{eq:addnoise}) ensures that an adversary can not easily determine if the data he observes is real, thus preserving the privacy of raw data. Moreover, we can control the amount of noise to achieve any desired level of privacy guarantee.

In fact, Lemma \ref{the:g-dp} and Corollary \ref{co:dp} reveal that the degree of differential privacy is directly related to the noise variance $\noisesigma$: larger noise results in enhanced privacy protection. The sensitivity $\sens$ is determined by the distribution system and can be approximated using domain expertise. For instance, in power grid analysis, the sensitivity of voltage data can be computed based on its standard operational range (ranging from $0\ p.u.$ to $1.1\ p.u.$).

\vspace{-1em}
\subsection{Quantification of Detection Performance Degradation}
\label{sec:degradation}
While (\ref{eq:addnoise}) enhances privacy protection, it may degrade the ability to detect line outages, potentially leading to increased detection delays and a higher false alarm rate. Therefore, it is crucial to analyze the extent to which detection performance is compromised when utilizing the encrypted data $\nx^{1:N}$. Only after completing this analysis can we devise a new solution to mitigate the degradation.

To study the performance degradation, we first note that the encrypted data $\nx[n]$ follows Gaussian distribution due to our choice of independent Gaussian noise for the raw data. Specifically, $\nx[n]$ follows $g_e\sim\mathcal{N}(\premu,\presigma+\mnoisesigma)$ before the outage ($n<\lambda$) and follows $f_e\sim\mathcal{N}(\postmu,\postsigma+\mnoisesigma)$ after the outage ($n\ge\lambda$). We use the notation $g_e$ and $f_e$ to denote the ``encrypted'' distributions, which are the results of introducing independent noise to distributions $g$ and $f$, respectively. 

Having defined $g_e$ and $f_e$, we can now rigorously measure the performance degradation. In Theorem \ref{the:noiseimpact}, we demonstrate that the ``distance'' between $g_e$ and $f_e$ is smaller than that between $g$ and $f$ by evaluating their Kullback-Leibler (KL) divergence. The ``closer'' the distributions are, the more challenging it is to distinguish them in the outage detection procedure, thus leading to a prolonged detection delay. Intuitively, if the noise term is infinitely large ($\noisesigma\to\infty$), the distributions $g_e$ and $f_e$ will be dominated by the same noise distribution and become impossible to distinguish.

\begin{theorem}\label{the:noiseimpact}
The randomizing scheme (\ref{eq:addnoise}) diminishes the KL divergence between pre- and post-outage distributions:
\begin{align}
    &\text{KL}_{\Delta} := D_{\text{KL}}(f\|g)-D_{\text{KL}}(f_e\|g_e) \ge 0, \label{eq:klreduction}\\
    &\text{KL}_{\Delta} \le \mathcal{O}(\noisesigma)(\|\premu-\postmu\|_2^2  + \frac{\left(\text{tr}(\postsigma)-\text{tr}(\presigma)\right)^2}{\text{tr}(\postsigma)}). \label{eq:upperbound}
\end{align}
\end{theorem}

\begin{proof}
For showing $\text{KL}_{\Delta}\ge 0$, we have
\begin{align*}
    &2\text{KL}_{\Delta}=\frac12(\premu-\postmu)^T[(\presigma)^{-1}-(\presigma^e)^{-1}](\premu-\postmu)\\
    &+\frac12\log\frac{|\presigma|}{|\postsigma|}\frac{|\postsigma^e|}{|\presigma^e|} + \frac12\text{tr}\{(\presigma)^{-1}(\postsigma) - (\presigma^e)^{-1}(\postsigma^e)\}\\
    &\ge \frac12{\displaystyle\sum}_{i=1}^{\dime} \left[ (\nu_i - \log \nu_i) - (\xi_i - \log \xi_i) \right],
\end{align*}
where $\boldsymbol{\Sigma}_i^e=\boldsymbol{\Sigma}_i+\mnoisesigma$ for $i=0,1$. $\nu_1,\cdots,\nu_\dime$ and $\xi_1,\cdots,\xi_\dime$ are the eigenvalues of $(\presigma)^{-1}\postsigma$ and $(\presigma^e)^{-1}\postsigma^e$, respectively. The inequality is due to that matrix $(\presigma)^{-1}-(\presigma^e)^{-1}$ is positive semi-definite. Moreover, since $|\xi_i-1| \le |\nu_i-1|, \forall i=1,\cdots,\dime$, we finally obtain $\text{KL}_{\Delta}\ge0$. Aside from the lower bound as zero, an upper bound of $\text{KL}_{\Delta}$ is further derived as
    \begin{align*}
    &\text{KL}_{\Delta} \le \frac12  \|\premu-\postmu\|_2^2(\frac{1}{\eigenpre^{\min}} - \frac{1}{\eigenpre^{\min}+\noisesigma} ) \\
    &+\frac{M}{2} (\frac{\eigenpost^{\max}}{\eigenpre^{\min}} - \log\frac{\eigenpost^{\max}}{\eigenpre^{\min}} +\log\frac{\eigenpost^{\max}+\noisesigma}{\eigenpre^{\min}+\noisesigma} -\frac{\eigenpost^{\max}+\noisesigma}{\eigenpre^{\min}+\noisesigma})\\
    &\le \frac{\noisesigma}{2(\eigenpre^{\min})^2}(\|\premu-\postmu\|_2^2 + M\frac{(\eigenpost^{\max} - \eigenpre^{\min})^2}{\eigenpost^{\max}}),
    \end{align*}
where $\eigenpre^{\min}$ is the smallest eigenvalue of $\presigma$, and $\eigenpost^{\max}$ is the largest eigenvalue of $\postsigma$. 
\end{proof}

As a corollary of $D_{\text{KL}}(f_e\|g_e)\le D_{\text{KL}}(f\|g)$ in Theorem \ref{the:noiseimpact}, the asymptotic lower bound of average detection delay in (\ref{eq:lowerbound}) is increased when the randomizing scheme is applied:
$$
\ADD(\Lambda,f_e,g_e) \ge \ADD(\Lambda,f,g),
$$
resulting in a prolonged detection delay of finding the outage time given encrypted data $\nx^{1:N}$. Theorem \ref{the:noiseimpact} not only indicates a strict performance degradation but also infers the magnitude of this degradation by deriving the upper bound of $\text{KL}_{\Delta}$. That is, we know approximately how much extra delay is brought w.r.t. the noise variance $\noisesigma$.

To illustrate the prolongation of detection delay, we present Fig. \ref{fig:statistic}, comparing two scenarios: the application of the statistic (\ref{eq:statistic}) to raw data $\Lambda(\x^{1:N})$ (red curve) and its application to encrypted data $\Lambda(\nx^{1:N})$ (blue curve). Due to the KL divergence reduction established in Theorem \ref{the:noiseimpact}, $\Lambda(\nx^{1:N})$ is typically smaller than $\Lambda(\x^{1:N})$ (we will show this claim later in the paper), especially after the outage occurrence. This inequality has two intuitive consequences. Firstly, it reduces the likelihood of triggering a false alarm when detecting the outage time using encrypted data, i.e., $\FAR(\tilde{\Lambda},f_e,g_e) \le \FAR(\Lambda,f,g)\le\alpha$. Secondly, encrypted data leads to a prolonged detection delay, i.e., the performance degradation.

To address the performance degradation, as suggested by Fig. \ref{fig:statistic}, a logical approach is to design a new detection statistic (represented by a green curve) to process the encrypted data. The new detection procedure is expected to maintain a comparable detection delay to the optimal scenario with access to raw data and still restrict the false alarm rate below $\alpha$.

\begin{figure}[ht]
\centering
\vskip -0.1in
\includegraphics[width=0.9\linewidth]{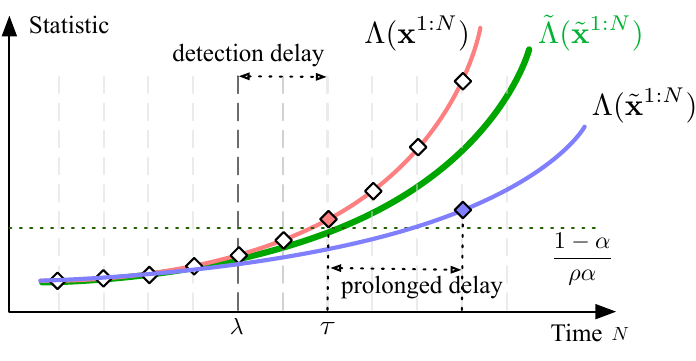}
\vskip -0.15in
\caption{Outages are reported when the calculated statistic surpasses the threshold $\frac{1-\alpha}{\rho\alpha}$. See Table \ref{tab:statistic} for a summary of these statistics.}
\vskip -0.1in
\label{fig:statistic}
\end{figure}

\vspace{-1em}
\subsection{A New Statistic to Boost the Detection Performance}
\label{sec:statistic}

In this subsection, we formally introduce a noise-mitigation technique to achieve detection performance comparable to the optimal scenario with access to raw data. We term it ``noise-mitigation'' since the technique essentially alleviates the performance impact resulting from the privacy-protective noise. To achieve this, we design a new statistic $\tilde{\Lambda}$ to process the encrypted data $\nx^{1:N}$. The new statistic aims to offer an approximation of the optimal statistic $\Lambda(\x^{1:N})$ (as depicted by the green curve in Fig. \ref{fig:statistic}), even when raw data is not available. We refer to the statistic $\Lambda(\x^{1:N})$ as ``optimal'' due to its demonstrated optimal detection performance when raw data is available (see Theorem \ref{theorem:Bayesian}). It's also important to note that this optimal statistic doesn't incorporate privacy protection. To prevent any ambiguity with these statistics, we present Table \ref{tab:statistic} for a comprehensive summary of the detection statistics used in this paper, along with their relevant attributes.

\begin{table}[ht]
\centering
\vskip -0.15in
\caption{Summarize of detection statistics}\label{tab:statistic}
\vskip -0.1in
\begin{tabular}{lccc}
    \toprule
    Statistic & Calculation & Privacy & Detection Performance \\
    \midrule
    $\Lambda(\x^{1:N})$ & (\ref{eq:statistic})  & $\times$  & optimal\\
    $\Lambda(\nx^{1:N})$ & (\ref{eq:statistic})  & $\checkmark$ & compromised\\
    $\tilde{\Lambda}(\nx^{1:N})$ & (\ref{eq:post_new}) &  $\checkmark$ & sub-optimal\\
    $\tilde{\Lambda}_{\gamma}(\nx^{1:N})$ & (\ref{eq:post_biased}) & $\checkmark$ & optimal \\
    \bottomrule 
\end{tabular}
\vskip -0.1in
\end{table}

For designing a new statistic $\tilde{\Lambda}$ that approximates $\Lambda$, we leverage the following insights. While the noise is generated randomly, its pattern, specifically the distribution parameters $\noisesigma$, are known to utility operators. This insight prompts us to compute the expectation of noise-related terms in the statistic $\Lambda(\nx^{1:N})$. By replacing these terms with their respective expectations, we can provide an unbiased estimation of $\Lambda(\nx^{1:N})$. Following this rationale, the new design for the statistic $\tilde{\Lambda}(\nx^{1:N})$ is presented in (\ref{eq:post_new}):
\begin{align}\label{eq:post_new}
\tilde{\Lambda}(\nx^{1:N}) = \sum_{k=1}^N \pi_N^k \prod_{n=k}^{N}\frac{\sqrt{|\boldsymbol{\Sigma}_0|}\exp\left(\beta_1[n]\right)}{\sqrt{|\boldsymbol{\Sigma}_1|}\exp\left(\beta_0[n]\right)},
\end{align}
where $\beta_i[n]:=-\frac{1}{2}(\nx[n]-\boldsymbol{\mu}_i)^T(\boldsymbol{\Sigma}_i)^{-1}(\nx[n]-\boldsymbol{\mu}_i) +\frac{1}{2} \noisedev\cdot\text{tr}(\boldsymbol{\Sigma}_i^{-1})$ for $i=0,1$. We note that $\beta_i[n]$ is an unbiased estimation of the corresponding term in the optimal statistic $\Lambda(\x^{1:N})$, i.e.,
$
\mathbb{E}_{\noise\sim\mathcal{N}(\boldsymbol{0},\mnoisesigma)} \beta_i[n] = -\frac{1}{2}(\x[n]-\boldsymbol{\mu}_i)^T(\boldsymbol{\Sigma}_i)^{-1}(\x[n]-\boldsymbol{\mu}_i)
$. By the unbiased design, the proposed statistic $\tilde{\Lambda}(\nx^{1:N})$ serves as the desired approximation of the optimal statistic $\Lambda(\x^{1:N})$. This effect is shown in Fig. \ref{fig:statistic} and proved in Lemma \ref{lemma:statistic}.

\begin{lemma}\label{lemma:statistic}
    The proposed statistic $\tilde{\Lambda}$ in (\ref{eq:post_new}) satisfies 
    \begin{align}\label{eq:statisticinequality}
        \Lambda(\nx^{1:N}) \le \tilde{\Lambda}(\nx^{1:N}) \le \Lambda(\x^{1:N}), \quad N\ge \lambda.
    \end{align}
\end{lemma}
\begin{proof}
For showing $\Lambda(\nx^{1:N}) \le \tilde{\Lambda}(\nx^{1:N})$, it suffices to show $\frac{f_e(\nx[n])}{g_e(\nx[n])} \le |\postsigma|^{\frac{1}{2}}/|\presigma|^{\frac{1}{2}}\exp(\beta_1[n]-\beta_0[n])$, where these two terms are denoted as $(\#^e)$ and $(\ast)$. In fact, we have
\begin{align*}
\log (\#^e) &= \frac{\dime}{2}\log\frac{s_0+\noisesigma}{s_1+\noisesigma} + (\frac{a_0}{s_0+\noisesigma} - \frac{a_1}{s_1+\noisesigma}),\\
\log (\ast) &= \frac{\dime}{2} \log\frac{s_0}{s_1} + (\frac{a_0}{s_0} - \frac{a_1}{s_1}) - \frac{\dime \noisesigma}{2} (\frac{1}{s_0} - \frac{1}{s_1}),
\end{align*}
where $a_0 = \frac{1}{2}\|\nx[n]-\boldsymbol{\mu}_0\|^2$ and $a_1 = \frac{1}{2}\|\nx[n]-\boldsymbol{\mu}_1\|^2$, and we consider diagonal co-variances $\boldsymbol{\Sigma}_0 = \text{diag}(s_0,\cdots,s_0)$ and $\boldsymbol{\Sigma}_1 = \text{diag}(s_1,\cdots,s_1)$. When $n\ge\lambda$ (the line outage occurs), we have $a_0\gg a_1$, which results in $\log (\#^e) \le \log (\ast)$.

For showing $\tilde{\Lambda}(\nx^{1:N}) \le \Lambda(\x^{1:N})$, it suffices to show $\frac{f(\x[n])}{g(\x[n])} \ge |\postsigma|^{\frac{1}{2}}/|\presigma|^{\frac{1}{2}}\exp(\beta_1[n]-\beta_0[n])$, where these two terms are denoted as $(\#)$ and $(\ast)$. In fact, we have
\begin{align*}
    2\log \frac{(\#)}{(\ast)} = \dime(\noisesigma-\|\noise[n]\|^2) (\frac{1}{s_0} - \frac{1}{s_1}) + 2\frac{b_1}{s_1} - 2\frac{b_0}{s_0},
\end{align*}
where $b_0 = \langle \noise[n], \boldsymbol{\mu}_0\rangle$ and $b_1 = \langle \noise[n], \boldsymbol{\mu}_1\rangle$ satisfying $\mathbb{E}_{\noise} b_0 = \mathbb{E}_{\noise} b_1 = 0$. It indicates that $\mathbb{E}_{\noise} [\log \frac{(\#)}{(\ast)}] = 0$, which further gives us $\mathbb{E}_{\noise} [\frac{(\#)}{(\ast)}] \ge \exp(0)=1$ from Jensen's inequality. It implies a higher likelihood that $(\#)>(\ast)$. Since at every time step $n=k$ we will randomly generate a noise vector $\noise[n]$, we can conclude that $\tilde{\Lambda}(\nx^{1:N}) \le \Lambda(\x^{1:N})$.
\end{proof}

From Lemma \ref{lemma:statistic}, the proposed statistic $\tilde{\Lambda}(\nx^{1:N})$ falls between the $\Lambda(\nx^{1:N})$ and $\Lambda(\x^{1:N})$ after the outage event ($n\ge\lambda$), aligning with our expectations in Fig. \ref{fig:statistic}. Consequently, it will exhibit a reduced false alarm rate (FAR) compared to the optimal scenario with raw data access, and will help alleviate the prolongation of average detection delay (ADD).

\begin{corollary}\label{co:performance}
    The proposed $\tilde{\Lambda}$ in (\ref{eq:post_new}) restricts the FAR below $\alpha$, and alleviate the prolongation of ADD, i.e.,
    \begin{align}\label{eq:performance}
    \FAR(\Lambda,f_e,g_e)& \le \FAR(\tilde{\Lambda},f_e,g_e) \le \FAR(\Lambda,f,g)\le \alpha\nonumber\\
    \ADD(\Lambda,f_e,g_e) &\ge \ADD(\tilde{\Lambda},f_e,g_e) \ge \ADD(\Lambda,f,g),
    \end{align}
\end{corollary}

As indicated by the proof of Lemma \ref{lemma:statistic}, Jensen's inequality hinders the attainment of a ``perfect'' approximation to the optimal statistic $\Lambda$, resulting in a remaining gap between $\tilde{\Lambda}$ and $\Lambda$. To address this matter, a logical approach is to seek specific conditions under which Jensen's inequality converges toward equality. With this in mind, we modify the statistic in (\ref{eq:post_new}) by introducing a constant term $\gamma\ge1$ as
\begin{align}\label{eq:post_biased}
\tilde{\Lambda}_\gamma(\nx^{1:N}) = \sum_{k=1}^N \pi_N^k \prod_{n=k}^{N}\frac{\sqrt{|\boldsymbol{\Sigma}_0|}\exp\left(\beta_1[n]/\gamma\right)}{\sqrt{|\boldsymbol{\Sigma}_1|}\exp\left(\beta_0[n]/\gamma\right)}.
\end{align}

We refer to the constant term $\gamma$ as the variance scaling factor since it scales the variance of term $\beta_i[n]$ by $1/\gamma^2$ times. When $\gamma=1$, the statistic in (\ref{eq:post_biased}) degrades to the statistic in (\ref{eq:post_new}). We employ this variance scaling factor because Jensen's inequality tends to become equality as the variance of variable approaches zero. To describe the effect of introducing $\gamma$ to scale $\beta_i$, we provide Lemma \ref{lemma:variance}. From previous discussing, the term $\beta_i$ is an unbiased estimation of $\bar{\beta}_i:=-\frac{1}{2}(\x[n]-\boldsymbol{\mu}_i)^T(\boldsymbol{\Sigma}_i)^{-1}(\x[n]-\boldsymbol{\mu}_i)$, whose variance is denoted as $\sigma_i^2$. Thus, the scaled term $\beta_i/\gamma$ used in (\ref{eq:post_biased}) can be modeled in a distribution $\mathcal{P}$ with mean $\bar{\beta}_i/\gamma$ and variance $\sigma_i^2/\gamma^2$. According to the theorem of the Jensen inequality gap, we have the following upper bound w.r.t. to the variance scaling factor $\gamma$.

\begin{lemma}\label{lemma:variance}
    Suppose $|\exp(\beta_i/\gamma)-\exp(\bar{\beta}_i/\gamma)|\le M|\beta_i/\gamma-\bar{\beta}_i/\gamma|^2$ for some $M$ and any $\beta_i/\gamma\in\mathbb{R}$, for any convex function $f$, we have an upper bound of Jensen gap as
    $$
    [\mathbb{E}[f(\frac{\beta_i}{\gamma})] - f(\mathbb{E}[\frac{\beta_i}{\gamma}])] \le M\int |\frac{\beta_i}{\gamma}-\frac{\bar{\beta}_i}{\gamma}|^2 d\mathcal{P}(\frac{\beta_i}{\gamma}) \le M \frac{\sigma_i^2}{\gamma^2}.
    $$
\end{lemma}

According to Lemma \ref{lemma:variance}, the variance-reduction technique in (\ref{eq:post_biased}) can narrow the gap in the Jensen inequality, consequently achieving a nearly perfect approximation of the optimal statistic. In summary, when implementing the randomization scheme (\ref{eq:addnoise}) to encrypt raw data and utilizing the new statistic (\ref{eq:post_biased}) for outage detection, we outline the privacy-aware line outage detection procedure, referred to as {\bf \Abbe}, in Algorithm \ref{alg:algorithm}. The proposed \Abbe offers two key advantages. First, it ensures privacy preservation by using noise for data encryption. Second, the proposed statistic provides an approximation to the optimal statistic when raw data is accessible, thereby achieving a comparable lower bound on detection delay while limiting the false alarm rate to a predefined tolerance level.

\begin{algorithm}
    \caption{Privacy-aware Line Outage Detection (\Abbe) with Boosted Detection Performance.}
    \label{alg:algorithm}
    \begin{algorithmic}[1] 
    \STATE \textbf{Input}: New voltage data $\x[n]$\\
    \STATE \textbf{Parameter}: Noise variance $\noisesigma$, variance scaling factor $\gamma$\\
    \STATE \textbf{Output}: Outage time
        \STATE Apply {\bf noise} to encrypt raw data.
        $$\nx[n] = \x[n] + \noise[n],\ \noise[n]\sim\mathcal{N}(\boldsymbol{0},\text{diag}(\noisesigma,\cdots,\noisesigma))$$
        \STATE Calculate detection {\bf statistic} $\tilde{\Lambda}(\nx^{1:n})$ in (\ref{eq:post_biased}).\\
        \IF {$\tilde{\Lambda}_{\gamma}(\nx^{1:n}) \geq \frac{1-\alpha}{\rho\alpha}$}
        \FOR{$i,k\in\grid$}
        \IF {$|\rho_{ik}^{-}| > \delta_{\max}$ and $|\rho_{ik}^{+}| < \delta_{\min}$}
        \STATE  {\bf report} outage time $\tau = N$ and the out-of-service branch between bus $i$ and $k$
        \ENDIF
        \ENDFOR
        \ENDIF
    \end{algorithmic}
\end{algorithm}

\vspace{-1em}
\section{Validation On Extensive Outage Scenarios With
Real-world Data}
\label{sec:simulation}
This section evaluates the privacy guarantee, the average detection delay, and the false alarm rate of \Abbe, comparing it with recent baselines on privacy-aware detection methods.

{\bf Dataset configuration.}\ To assess \Abbe across diverse system sizes and environments, we conduct comprehensive experiments using various network configurations. The systems include the IEEE 8-bus and IEEE 123-bus networks \cite{kersting1991radial}, along with two representative European distribution systems: a medium voltage (MV) network in an urban area and a low voltage (LV) network in a suburban area \cite{mateo2018european}. We utilize these two networks from Europe to contrast with standard IEEE bus networks typical of the U.S. for the consideration of  diverse grid architectures. In these two networks, we still focus on using the network's topology to simulate customer-level voltage data for outage detection. In each of these networks, we select bus 1 as the slack bus.

In recognition of the complexities in real-world distribution grid outage scenarios, we explore situations where alternative power sources come into play following a line outage. In such scenarios, relying solely on the ``last gasp'' notification becomes less effective, rendering the detection of line outages more challenging. To model this complexity, we conduct simulations for the following two representative scenarios.

\begin{itemize}[leftmargin=*]
    \item \underline{Mesh networks}. Mesh networks are often used to model networks in urban areas, where most buses retain non-zero voltages after a line outage as they can receive power from alternative branches. To simulate mesh networks, we introduce loops into the aforementioned systems, ensuring their connectivity remains intact after line outages \cite{liao2021quick}. As an example, in the IEEE-123 bus network, we introduce loops by adding two branches: one between bus 77 and 120 and another between bus 50 and 56, with admittances matching that of the branch between bus 122 and 123.
    
    \item \underline{Radial networks with DERs}. In such case, some buses continue to receive power from DERs though isolated from the main grid after a line outage. This type of outage scenario is typical in residential areas. To simulate DERs, we select multiple buses to have solar power panels with batteries as energy storage. For solar panels, we use power generation profiles computed using the PVWatts Calculator \cite{dobos2014pvwatts}.
\end{itemize}

To generate more authentic data, we use real residential power profiles from the Duquesne Light Company (DLC) in Pittsburgh, USA. The DLC dataset comprises anonymized and secure hourly (and 15-minute) smart meter readings of active power from over 5,000 houses throughout the year 2016. Basic statistics of this dataset are provided in Table \ref{tab:DLC}.
\begin{table}[H]
    \centering
    \vskip -0.15in
    \caption{Statistical analysis of DLC power dataset.}
    \label{tab:DLC}
    \vskip -0.1in
    \begin{tabular}{c|c}
    \toprule
    Statistics & Value \\
    \midrule
    Minimum Value        & $-2.6040$    \\
    Maximum Value      & $26.6860$   \\
    Mean     & $0.8473$    \\
    Standard Deviation   & $0.6387$    \\
    Skewness & $1.7441$\\
    \midrule
    \end{tabular}
    \vskip -0.10in
\end{table}

{\bf Implementing details.}\ The time-series voltage magnitude data are generated using the MATLAB Power System Simulation Package (MATPOWER) in MATLAB R2022b. In each distribution system, we assign active power $p_i[n]$ from the DLC power profile to  bus $i$ at time $n$. The reactive power $q_i[n]$ is determined based on a randomly generated power factor $pf_i[n]$, which follows a uniform distribution $\text{Unif}(0.9,1)$. Using the active and reactive power values, we employ MATPOWER to solve power flow equations and derive voltage measurements. Additionally, we simulate outage scenarios by setting the admittance of one or multiple branches to zero and solve the power flow equation again.

For more robust evaluation, each experiment is conducted using Monte Carlo simulation with over 1000 replications, where the voltage sequence in (\ref{eq:sequence}) is generated by concatenating $\lambda-1$ records from pre-outage data and 50 records from post-outage data (50 samples are sufficient since the detection delay in our experiments is lower than 50). The outage time $\lambda$ is randomly generated using a geometric distribution $\text{Geo}(\rho)$. This geometric prior is based on our belief that outages can occur independently at any time step, with an equal probability of $\rho$. We choose $\rho=0.04$ in our experiments, which is derived from historical outage data, indicating that each time step has a 4\% chance of experiencing a line outage. Another threshold $\alpha$
 is set at $10^{-2}$, selected through a cross-validation process that balances statistical analysis with operational needs, ensuring optimal trade-off between detection delay and false alarm rates in the aforementioned grid systems.

After obtaining voltage data from MATLAB, the remaining calculations for outage detection in Algorithm \ref{alg:algorithm} are implemented using Python 3.8 on a personal computer with a Windows 10 operating system, an Intel Core i7 processor clocked at 2.2 GHz, and 16 GB of RAM.

{\bf Baseline methods.}\ In the following experiments, the optimal Bayesian detection procedure with access to raw data ($\Lambda(\x^{1:N})$) is referred to as {\bf benchmark}. It should have optimal detection performance but has no privacy guarantee.
The same detection statistic applied to encrypted data ($\Lambda(\nx^{1:N})$) is referred to as {\bf privacy-only} since it degrades the detection performance. To remove the performance degradation, our proposed method ($\tilde{\Lambda}_{\gamma}(\nx^{1:N})$) in Algorithm \ref{alg:algorithm} is referred to as {\bf \Abbe}.
The noise level $\noisesigma$ and the variance scaling factor $\gamma$ will be further pointed out.
We also compare with recent techniques dealing with privacy concerns, including a private approximation of the change-point maximum likelihood estimation {\bf MLE} \cite{cummings2018differentially}, an uncertain likelihood ratio {\bf ULR} proposed to replace the original detection statistic \cite{hare2021toward}, a sanitize channel method {\bf SCM} \cite{lau2020privacy} that adheres to the information leakage requirement \cite{issa2016operational}, and a private stream aggregation {\bf PSA} \cite{kurt2022online} approach to meet the differential privacy.

\vspace{-1em}
\subsection{Visualization of Privacy Guarantee}
Before evaluating the detection performance of our privacy-aware approach, we first visualize how much privacy is preserved when using DLC data in the IEEE 8-bus system.
To achieve so, we plot the corresponding trade-off function within the Gaussian differential privacy framework \cite{dong2019gaussian} against baseline trade-off functions.
Specifically, as established in Theorem \ref{the:g-dp}, applying our privacy-protection scheme (\ref{eq:addnoise}) yields $G_\frac{\sens}{\noisedev}$- Gaussian differential privacy, and produces corresponding trade-off functions $T(\mathcal{N}(0,1), \mathcal{N}(\frac{\sens}{\noisedev},1))$ at varying levels of noise variance $\noisesigma$. Therefore, we can compare these trade-off functions with baseline trade-off functions $G_\mu:=T(\mathcal{N}(0,1), \mathcal{N}(\mu,1)), \mu=0.5, 1, 3$. This comparison can tell us the difficulty an attacker would encounter in compromising user privacy under our scheme, quantitatively demonstrating how much privacy is preserved.

The results are shown in Fig. \ref{fig:privacyvisualize}, highlighting our method's capability to obscure the distinction between encrypted and raw data, significantly enhancing privacy. With a noise variance of $\noisesigma=5e-3$, differentiating the encrypted data $\nx^{1:N}$ from the raw data $\x^{1:N}$ for an attacker becomes more difficult than distinguishing distributions $\mathcal{N}(0,1)$ and $\mathcal{N}(3,1)$. Increasing the noise variance to $\noisesigma=4e-2$ further intensifies this effect, elevating privacy protection to levels where distinguishing encrypted from genuine data becomes more challenging than differentiating between $\mathcal{N}(0,1)$ and $\mathcal{N}(1,1)$.

\begin{figure}[ht]
\centering
\includegraphics[width=0.9\linewidth]{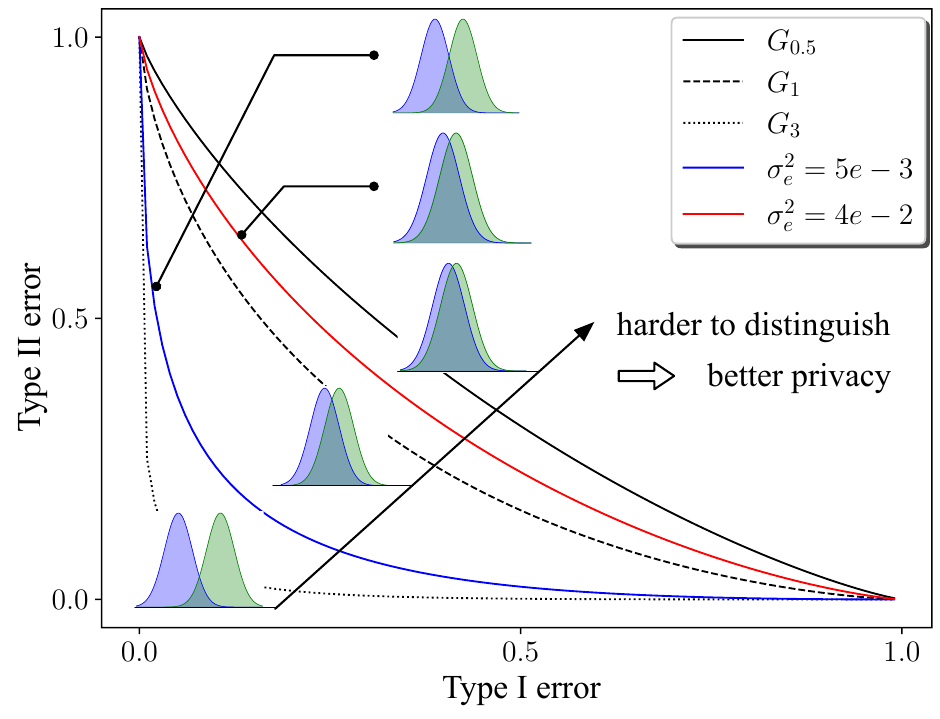}
\vskip -0.15in
\caption{The comparison of trade-off functions of distinguishing unit-variance Gaussian distributions using DLC data and IEEE 8-bus system simulation.}
\vskip -0.15in
\label{fig:privacyvisualize}
\end{figure}

\subsection{Evaluation of the the Noise-mitigation Design}
Despite the privacy protection shown in Fig. \ref{fig:privacyvisualize}, applying the randomizing scheme (\ref{eq:addnoise}) will inevitably lead to a decline in detection performance (see Section \ref{sec:degradation}). Within this subsection, we assess whether our proposed detection procedure \Abbe can counteract this performance degradation. 

To achieve so, we compare between the optimal detection statistic (\ref{eq:statistic}), which lacks privacy protection, and our innovative statistic (\ref{eq:post_new}) with privacy protection. Their logarithmic values are illustrated in Fig. \ref{fig:noise-mitigation}, where the simulation is performed in the IEEE 123-bus system selected specifically to examine the efficacy of \Abbe in a large-scale network. As we can see, the optimal statistic $\Lambda(\x^{1:N})$ (red) increases dramatically after the outage time $\lambda=30$, resulting in a near-zero detection delay. The same statistic applied to encrypted data $\Lambda(\nx^{1:N})$ has a privacy guarantee but postpones the detection (blue). Our proposed statistic $\tilde{\Lambda}(\nx^{1:N})$ (green) in (\ref{eq:post_new}) closely approximates the optimal statistic, thus effectively mitigating the postponing effect while still preserving privacy. This evaluation elucidates our method's dual achievement of maintaining privacy without compromising the timeliness and accuracy of outage detection. Meanwhile, it is worth noting that Fig. \ref{fig:noise-mitigation} serves to validate the conclusions drawn in Lemma \ref{lemma:statistic}.

\begin{figure}[ht]
\centering
\includegraphics[width=0.9\linewidth]{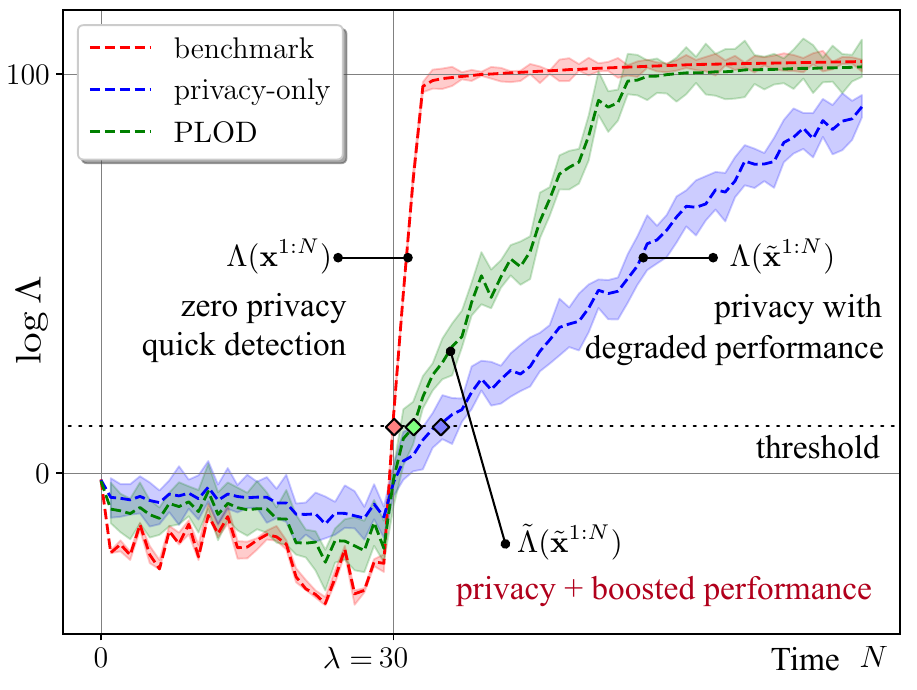}
\vskip -0.1in
\caption{The logarithm of various detection statistics in IEEE 123-bus system. $\lambda=30$, $\noisesigma=4e-2$, $\gamma=1$ and $\alpha=1\%$.}
\vskip -0.2in
\label{fig:noise-mitigation}
\end{figure}

\vspace{-1em}
\subsection{Evaluation of the Variance-reduction Design}
From Fig. \ref{fig:noise-mitigation}, even though our proposed statistic (green) serves to approximate the optimal statistic (red) to mitigate the detection performance degradation, there is still a gap due to the effect of Jensen's inequality. This subsection evaluates \Abbe's ability to fill this gap. Specifically, our newly designed approach in (\ref{eq:post_biased}) aims to reduce the statistic's variance, ultimately narrowing this gap. To confirm this improvement, we plot the logarithm of statistic $\tilde{\Lambda}_\gamma(\nx^{1:N})$ as defined in (\ref{eq:post_biased}) for various choices of the variance scaling factor $\gamma$, and compare them with the optimal statistic $\Lambda(\x^{1:N})$. The results are shown in Fig. \ref{fig:variancereduction}. As $\gamma$ increases from $1$ to $3$, $\tilde{\Lambda}_\gamma(\nx^{1:N})$ progressively converges to a more precise approximation of the optimal statistic $\Lambda(\x^{1:N})$. This convergence is a consequence of the reduced variance in the detection statistic, as shown in two zoomed-in illustrations in Fig. \ref{fig:variancereduction}. This reduction in variance effectively constrains the error gap in Jensen's inequality, as elaborated in Lemma \ref{lemma:variance}.

Our findings illustrate that by applying the adjusted statistic from (\ref{eq:post_biased}) within \Abbe, it's possible to closely mimic the optimal statistic derived from raw data $\x^{1:N}$, even when only encrypted data $\nx^{1:N}$ is available. This indicates \Abbe's potential to not just protect customer data privacy but also to sustain detection performance. The effectiveness of the detection performance in real-life outage dataset will be further explored in the following subsection.

\begin{figure}[ht]
\centering
\vskip -0.10in
\includegraphics[width=0.9\linewidth]{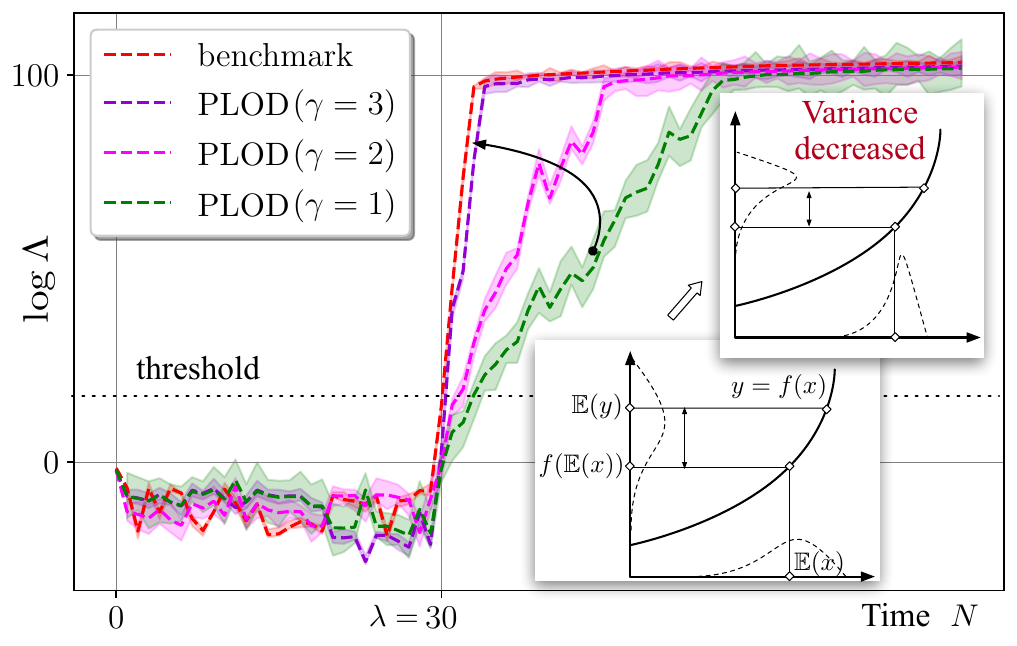}
\vskip -0.1in
\caption{The logarithm of detection statistics with different variance scaling factor $\gamma$ in the IEEE 123-bus system. $\lambda=30$, $\noisesigma=4e-2$, and $\alpha=1\%$.}
\vskip -0.1in
\label{fig:variancereduction}
\end{figure}

\begin{table*}[ht]
    \centering
    \caption{Performance Comparison on Various Systems. $\noisesigma=4e-2$ and $\alpha=1\%$.}
    \vskip -0.05in
    \begin{tabular}{cccc|cccc|cccc}
    \toprule
    System & \multicolumn{3}{c}{System information} & \multicolumn{4}{c}{Average Detection Delay ($1$ unit) $\downarrow$} & \multicolumn{4}{c}{Empirical False Alarm Rate $(\%)$ $\downarrow$} \\
    \midrule
    Mesh network & \#Branches & \#DERs & Outage  & bench. & MLE \cite{cummings2018differentially} & ULR \cite{hare2021toward} & {\bf \Abbe} &  bench. & MLE \cite{cummings2018differentially} & ULR \cite{hare2021toward} & {\bf \Abbe} \\ \midrule
    $8$-bus &  $9$  &$0$ & $4$-$7$                  & $3.10$ & $3.83^{+0.73}$  & $3.94^{+0.84}$ & $3.35^{+0.25}$ & $0.9$ & $4.5^{+3.6}$  & $4.2^{+3.3}$ & $1.2^{+0.3}$\\
    $8$-bus &  $9$  &$8$ & $4$-$7$                  & $3.82$ & $4.76^{+0.94}$ & $4.86^{+1.04}$ & $4.11^{+0.29}$ & $1.4$ & $3.8^{+2.4}$ & $3.8^{+2.4}$ & $1.6^{+0.2}$\\
    $123$-bus &  $124$  &$0$ & $73$-$74$            & $0.89$ & $1.11^{+0.22}$ & $1.09^{+0.2}$ & $0.94^{+0.05}$ & $0.6$ & $1.5^{+0.9}$ & $1.6^{+1.0}$ & $0.8^{+0.2}$\\
    $123$-bus &  $124$  &$0$ & $73$-$74$,$14$-$15$  & $0.88$ & $1.16^{+0.28}$ & $1.16^{+0.28}$ & $0.94^{+0.06}$ & $0.5$ & $1.4^{+0.9}$ & $1.4^{+0.9}$ & $0.7^{+0.2}$\\
    $123$-bus &  $124$  &$0$ & $5$ branches         & $0.84$ & $1.16^{+0.32}$ & $1.07^{+0.23}$ & $0.91^{+0.07}$ & $0.5$ & $1.4^{+0.9}$ & $1.3^{+0.8}$ & $0.7^{+0.2}$\\
    LV suburban & $129$   & $0$  & $26$-$95$        & $3.80$ & $5.21^{+1.41}$ & $4.88^{+1.08}$ & $4.11^{+0.31}$ & $1.5$ & $3.6^{+2.1}$ & $3.4^{+1.9}$ & $2.0^{+0.5}$\\
    LV suburban & $129$   & $30$  & $26$-$95$       & $3.89$ & $5.35^{+1.46}$ & $4.91^{+1.02}$ & $4.19^{+0.30}$ & $1.3$ & $3.1^{+1.8}$ & $3.0^{+1.7}$ & $1.7^{+0.4}$\\
    MV urban & $48$   &$0$   & $34$-$35$            & $0.83$ & $1.08^{+0.25}$ & $1.12^{+0.29}$ & $0.89^{+0.06}$ & $0.4$ & $1.0^{+0.6}$ & $1.0^{+0.6}$ & $0.5^{+0.1}$\\
    MV urban & $48$   &$7$   & $34$-$35$            & $1.45$ & $1.87^{+0.42}$ & $1.91^{+0.46}$ & $1.55^{+0.10}$ & $0.8$ & $1.9^{+1.1}$ & $2.2^{+1.4}$ & $1.1^{+0.3}$\\
    \midrule
    Radial network & \# Branches & \# DERs & Outage  & bench. & MLE \cite{cummings2018differentially} & ULR \cite{hare2021toward} & {\bf \Abbe} &  bench. & MLE \cite{cummings2018differentially} & ULR \cite{hare2021toward} & {\bf \Abbe} \\ \midrule
    $8$-bus &  $7$  &$8$ & $4$-$7$                  & $3.5$ & $4.88^{+1.38}$ & $4.41^{+0.91}$ & $3.8^{+0.3}$ & $0.8$ & $1.9^{+1.1}$ & $1.9^{+1.1}$ & $1.0^{+0.2}$\\
    $8$-bus &  $7$  &$8$ & $2$-$6$                  & $3.58$ & $4.93^{+1.35}$ & $4.61^{+1.03}$ & $3.84^{+0.26}$ & $0.9$ & $2.1^{+1.2}$ & $2.3^{+1.4}$ & $1.1^{+0.2}$\\
    $123$-bus &  $122$ &$12$ & $73$-$74$            & $1.29$ & $1.73^{+0.44}$ & $1.69^{+0.4}$ & $1.39^{+0.1}$ & $0.5$ & $1.3^{+0.8}$ & $1.3^{+0.8}$ & $0.6^{+0.1}$\\
    $123$-bus &  $122$ &$122$ & $73$-$74$           & $6.75$ & $9.22^{+2.47}$ & $8.63^{+1.88}$ & $7.31^{+0.56}$ & $2.8$ & $7.8^{+5.0}$ & $6.7^{+3.9}$ & $3.8^{+1.0}$\\
    LV suburban & $114$   & $30$  & $26$-$95$       & $3.25$ & $4.43^{+1.18}$ & $4.34^{+1.09}$ & $3.53^{+0.28}$ & $1.4$ & $3.3^{+1.9}$ & $3.8^{+2.4}$ & $1.8^{+0.4}$\\
    LV suburban & $114$   & $113$  & $26$-$95$      & $8.87$ & $11.76^{+2.89}$ & $11.64^{+2.77}$ & $9.58^{+0.71}$ & $4.0$ & $10.1^{+6.1}$ & $10.6^{+6.6}$ & $5.4^{+1.4}$\\
    MV urban & $38$   &$7$   & $34$-$35$            & $1.45$ & $1.96^{+0.51}$ & $1.93^{+0.48}$ & $1.57^{+0.12}$ & $0.7$ & $1.7^{+1.0}$ & $1.9^{+1.2}$ & $0.9^{+0.2}$\\
    MV urban & $38$   &$7$   & $23$-$35$            & $1.69$ & $2.3^{+0.61}$ & $2.15^{+0.46}$ & $1.82^{+0.13}$ & $1.1$ & $3.2^{+2.1}$ & $2.6^{+1.5}$ & $1.5^{+0.4}$\\
    \bottomrule
    \end{tabular}
    \label{tab:my_label}
    \vskip -0.1in
\end{table*}

\subsection{Evaluation of Detection Performance: Average Detection Delay and False Alarm Rate}
After verifying the effects of our various designs within the proposed method \Abbe, we finally evaluate its overall detection performance when using encrypted data $\nx^{1:N}$. This evaluation covers both the average detection delay and the false alarm rate, aiming to provide a holistic view of the system's performance under privacy-preserving conditions.
To validate the asymptotic optimality of the detection delay in Theorem \ref{theorem:Bayesian}, we plot in the upper half of Fig. \ref{fig:overall-syn} the average delay $\E(\tau-\lambda|\tau\ge\lambda)$ divided by $|\log\alpha|$ and the theoretical lower bound $-\log(1-\rho)+D_{\text{KL}}(f||g)$. 
This comparison highlights that the detection delay of both the benchmark and the \Abbe method can achieve the optimal lower bound asymptotically. Notably, this asymptotic convergence contrasts with the higher delays observed in privacy-only approach, emphasizing the effectiveness of \Abbe in balancing privacy with quick detection response.

The detection rule in Theorem \ref{theorem:Bayesian} is also expected to restrict the false alarm rate below a predefined threshold $\alpha$. To confirm this, we analyze the empirical false alarm rate $\P(\tau<\lambda)$ in comparison to $\alpha$, as depicted in the lower section of Fig. \ref{fig:overall-syn}. Notably, our method parallels the benchmark performance, consistently keeping the empirical false alarm rate beneath $\alpha$, particularly as $\alpha$ approaches zero. This outcome underscores the efficiency of our algorithm in detecting line outages with minimal false alarms, demonstrating its reliability even with the utilization of encrypted data, thereby reinforcing its practical applicability in maintaining system security while adhering to privacy constraints.

\begin{figure}[H]
\centering
\vskip -0.10in
\includegraphics[width=1\linewidth]{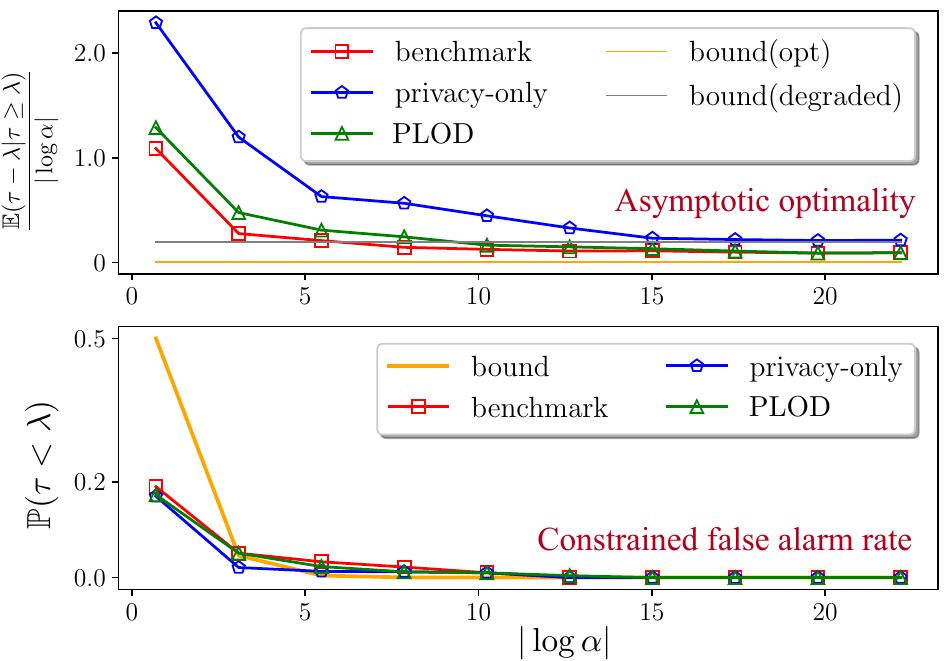}
\vskip -0.05in
\caption{The average detection delay (upper half) and the false alarm rate (lower half) in the IEEE 123-bus system. $\noisesigma=4e-2$, $\gamma=1$ and $\alpha=1\%$.}
\vskip -0.15in
\label{fig:overall-syn}
\end{figure}

To evaluate \Abbe across diverse grid systems under various outage configurations, we present a comprehensive summary of results in Table \ref{tab:my_label}. Throughout these experiments, we conduct comparisons not only with the benchmark method but also with two recent relevant techniques that provide privacy-aware approaches for detecting distribution changes. These methods are referred to as {\bf MLE} (Maximum Likelihood Estimation)  \cite{cummings2018differentially}, {\bf ULR} (Utility Learning-based Rule) \cite{hare2021toward}, {\bf SCM} (Sanitize Channel-based Method) \cite{lau2020privacy}, and {\bf PSA} (Private Stream Aggregation) \cite{kurt2022online}. To ensure consistency in the level of differential privacy guarantees across all methods, we apply noise with a variance of $\noisesigma=4e-2$ to the raw data.

Table \ref{tab:my_label} showcases \Abbe's ability to handle a variety of outage scenarios in both mesh (loopy) networks and radial networks with DERs penetration. Of particular note is the performance of \Abbe when no access to raw data is available, as it demonstrates shorter detection delays and lower false alarm rates compared to MLE and ULR. In contrast to the benchmark method, which has access to raw data and therefore carries a privacy risk, \Abbe exhibits only marginal degradation in detection delay and false alarm rate. Moreover, Table 1 reveals two significant observations. Firstly, in cases where multiple branches undergo simultaneous outages, the average detection delay tends to be shorter. This can be attributed to the increased Kullback-Leibler (KL) distance between the distributions $g$ and $f$ when multiple lines are disconnected. Secondly, in radial networks with a greater number of simulated DERs, the detection of line outages takes more time, primarily due to the smaller KL distance between $g$ and $f$ in such scenarios.

To assess detection performance under varying levels of noise introduced to the raw data, we present Table \ref{tab:outage}, which includes results for both the average detection delay (ADD) and false alarm rate (FAR). To ensure consistency in the level of differential privacy guarantees across all methods, we argue that implementing noise with a variance of $\noisesigma=9e-2$ in the \Abbe, MLE, and ULR techniques achieves a privacy level comparable to that provided by the SCM and SPA methods. The analysis from Table \ref{tab:outage} reveals that our method surpasses other privacy-preserving techniques in terms of both average detection delay and false alarm rate. This superior performance is attributed to our unique statistical design in (\ref{eq:post_new}) and (\ref{eq:post_biased}), which effectively mitigates the negative effects associated with the noise from privacy protection mechanisms. The results also reveal other method's characteristics. For example, the PSA method typically results in a longer detection delay and a lower false alarm rate. This could be due to its reliance on data aggregation techniques for privacy protection, which, while safeguarding data privacy, may inadvertently delay the timely detection of outages. Nonetheless, for operators prioritizing robust and consistent detection over detection speed, PSA presents a viable option.

\begin{table}[H]
\centering
\vskip -0.1in
\caption{Performance comparison at different noise level in the IEEE 8-bus system. $\alpha=1\%$ and $\gamma=3$.}\label{tab:outage}
\vskip -0.1in
\begin{tabular}{lccc}
    \toprule
    Noise & Method & ADD(unit) & FAR(\%) \\
    \midrule
    \multirow{ 3}{*}{$\noisesigma=0.01$} & \Abbe & ${\bf 2.71}\pm 0.53$ & ${\bf 1.03}\pm 0.02$\\
    & MLE \cite{cummings2018differentially} & $3.41\pm 0.72$ & $1.12\pm 0.15$\\
    & ULR  \cite{hare2021toward} & $3.80\pm 0.89$ & $1.57\pm 0.19$ \\
    \midrule
    \multirow{ 3}{*}{$\noisesigma=0.04$} & \Abbe & ${\bf 3.67}\pm 0.65$ & ${\bf 0.95}\pm 0.03$\\
    & MLE & $4.12\pm 0.83$ & $1.06\pm 0.37$\\
    & ULR  & $4.69\pm 1.01$ & $2.16\pm 0.74$ \\
    \midrule
    \multirow{ 3}{*}{$\noisesigma=0.09$} & \Abbe & ${\bf 4.56}\pm 0.77$ & ${\bf 3.82}\pm 7.39$ \\
    & MLE & $4.85\pm 0.81$ & $4.45\pm 0.81$ \\
    & ULR & $4.69\pm 1.01$ & $3.99\pm 0.74$ \\
    \midrule
    \multirow{ 2}{*}{} & SCM \cite{lau2020privacy}  & $4.64\pm 0.92$ & $4.02\pm 0.85$ \\
    & SPA \cite{kurt2022online} & $5.35\pm 1.30$ & $2.88\pm 0.19$ \\
    \bottomrule 
\end{tabular}
\vskip -0.05in
\end{table}

\subsection{Outage Branch Localization with Accuracy}
\label{sec:local_experiment}
Upon detecting an outage, we proceed to compute the conditional correlation between buses to pinpoint the out-of-service branch, as outlined in Theorem \ref{lemma:local}. Illustrated in Fig. \ref{fig:localization}, the absolute conditional correlation of every bus pair within the loopy 8-bus system is depicted both before and after a line outage at branch 4-7. Notably, the value within the red box transitions from a non-zero value pre-outage ($\rho^{-}_{47}>\delta_{\max}$) to near zero post-outage ($\rho^{+}_{47}<\delta_{\min}$), leading to the accurate localization of the out-of-service branch at 4-7, consistent with the ground truth. Furthermore, Fig. \ref{fig:localization}(d) demonstrates that the localization method using the covariance matrix through \Abbe with randomized data remains as effective as the benchmark scenario. This effectiveness likely stems from the fact that adding independent noise to the voltage data does not disrupt the condition where the voltage data of the two buses become conditionally independent following the line outage.

\begin{figure}[ht]
\centering
\vskip -0.15in
\includegraphics[width=1\linewidth]{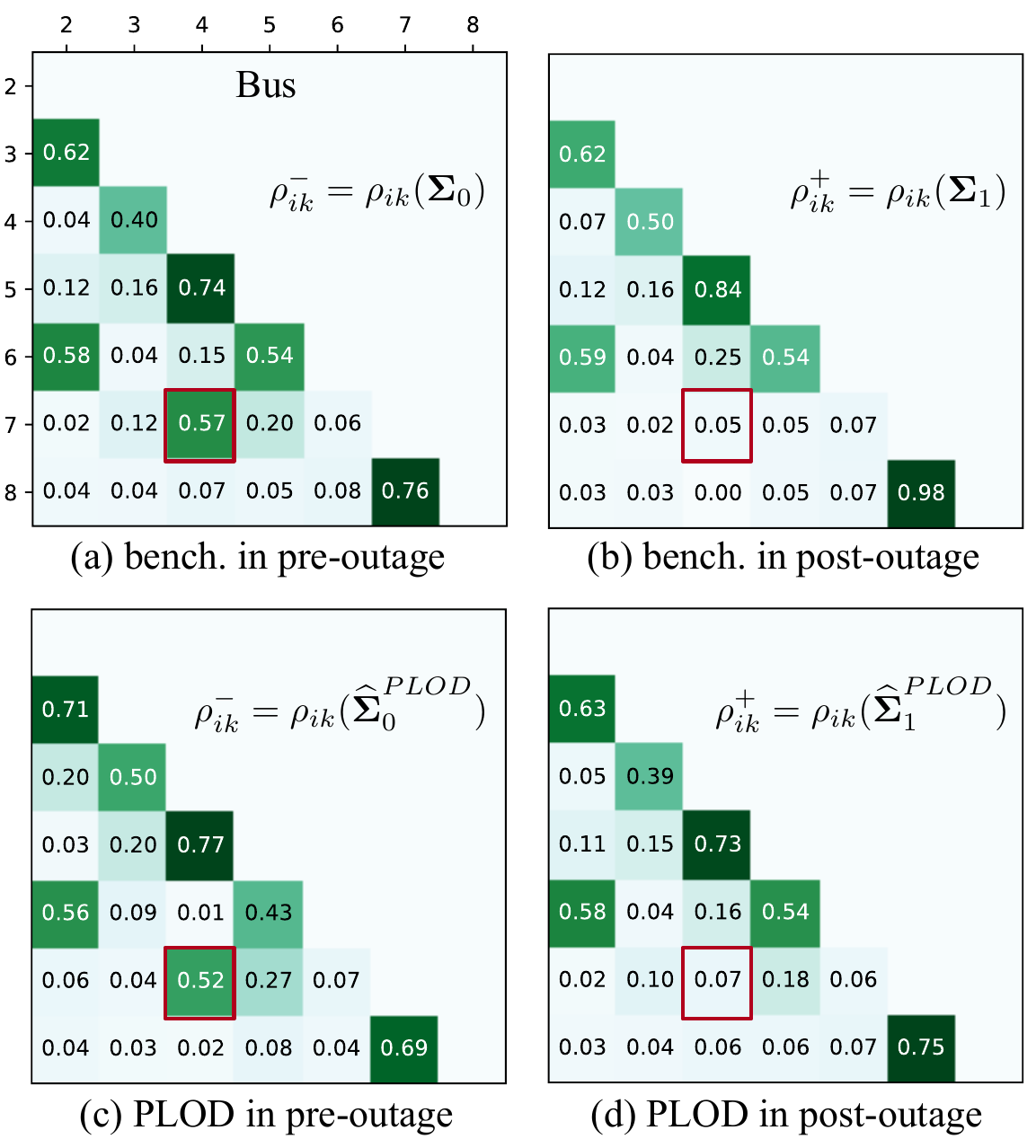}
\vskip -0.1in
\caption{Absolute conditional correlation of the loopy 8-bus system before and after an outage in branches 4-7. We choose $\delta_{\max}=0.5$ and $\delta_{\min}=0.1$.}
\vskip -0.1in
\label{fig:localization}
\end{figure}

Table \ref{tab:localize} displays the localization accuracy rates obtained from 1,000 experiments. It is clear that the proposed method \Abbe consistently achieves a localization accuracy of over 92\%, even when operating in privacy-aware scenarios where voltage data is subjected to randomized noise.

\begin{table}[H]
    \centering
    \vskip -0.15in
    \caption{Outage localization accuracy (\%) in algorithm \ref{alg:algorithm}. $\alpha=1\%, \delta_{\max}=0.5$ and $\delta_{\min}=0.1$.}
    \vskip -0.1in
    \begin{tabular}{c|c|cc}
    \toprule
    System & bench. & MLE & \Abbe \\
    \midrule
    $8$-bus         & $97.7$    & $92.9$ & $94.3$\\
    $123$-bus       & $95.0$    & $90.1$ & $92.7$ \\
    LV suburban     & $96.1$    & $91.4$ & $94.4$\\
    MV urban        & $95.8$    & $92.6$ & $93.8$ \\
    \midrule
    \end{tabular}
    \label{tab:localize}
    \vskip -0.2in
\end{table}

\vspace{-1em}
\subsection{Sensitivity Analysis to Data Coverage}
\label{sec:sensitivity}
In the distribution grid, access to the data of every bus is not guaranteed for several reasons. For instance, rural areas may lack smart meter installations, technical issues can result in data loss, and privacy concerns might lead to data refusal. Thus, an analysis of incomplete smart meter data coverage is necessary to assess \Abbe's real-world detection performance.

Based on records from \cite{USreport3}, over 107 million smart meters have covered 75\% of U.S. households by 2021. Therefore, we simulate the scenario where a fraction of buses (ranging from 75\% to 100\%) are randomly selected to provide voltage measurements for outage detection. The outcomes, illustrated in Fig. \ref{fig:coverage}, reveal how much additional Average Detection Delay (ADD) and False Alarm Rate (FAR) is introduced at various coverage ratios in comparison to 100\% data coverage. For instance, when applying \Abbe with variance scaling factors $\gamma$ equal to 1, 2, and 3, a 75\% data coverage ratio necessitates an additional 2.5, 1.9, and 1.7 data samples, respectively, to detect the outage. Simultaneously, the false alarm rate increases by 9.5\%, 11.9\%, and 13.1\%, respectively. It's noteworthy that when data coverage is not complete, increasing the variance scaling factor $\gamma$ involves a trade-off: it reduces detection delay at the expense of introducing more false alarms.

We note that our method doesn't rely on the assumption of 100\% sensor data coverage across the grid. In reality, power line outages tend to impact a majority of buses in the system, with the extent of impact varying based on their proximity to the outage location. It allows us to identify outages by detecting distribution changes in sensor data from some, rather than all, buses located near the source of the outage.

\begin{figure}[H]
\centering
\vskip -0.10in
\includegraphics[width=1\linewidth]{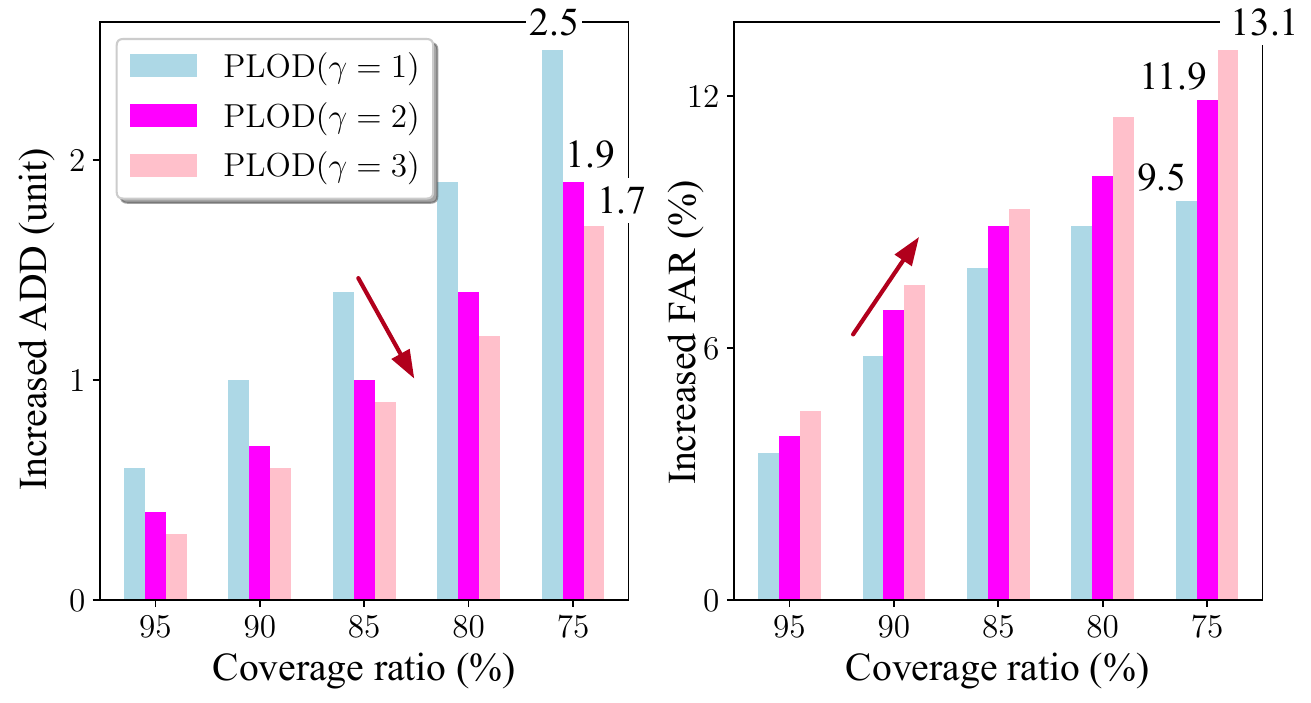}
\vskip -0.15in
\caption{Increased ADD (unit) and FAR (\%) under different ratios of data coverage compared to 100\% coverage in 123-bus loopy system. $\alpha=1\%$.}
\vskip -0.10in
\label{fig:coverage}
\end{figure}

\section{Conclusion}
\label{sec:conclusion}
In this paper, we introduce a robust privacy-aware method for detecting line outages in the distribution grids, effectively striking a balance between preserving privacy and maintaining detection performance. Our contributions encompass several key aspects. First, we ensure the direct protection of raw data through a randomization scheme embedded within the differential privacy framework. Second, we quantify the trade-off between privacy protection and detection performance, considering factors such as detection delay and false alarm rate. Lastly, we introduce a novel detection statistic that mitigates the adverse impact of encrypted data on detection performance, and in some cases, entirely eliminates it.

To validate our contributions, we conduct extensive experiments across a range of network systems, comprising 17 distinct outage configurations. The empirical results underscore the success of our privacy-aware outage detection methodology, achieving a harmonious blend of privacy preservation and detection performance that is comparable to the optimal case.

We also highlight that the proposed privacy-preserving technique, developed for outage detection, exhibits a versatile applicability extending beyond its initial scope to encompass various types of fault and anomaly detection challenges within power systems. This adaptability stems from the method's foundational approach to data protection, ensuring that sensitive information remains secure while still enabling the accurate identification of many types of anomalies. Consequently, whether addressing line faults, equipment malfunctions, or unexpected fluctuations in power distribution networks, our technique maintains its effectiveness, offering a robust solution for enhancing grid resilience and reliability without compromising customer privacy.

\vspace{-1em}
\bibliographystyle{IEEEtran}
\bibliography{reference}

\end{document}